\RequirePackage{fix-cm}
\documentclass{svjour3}                     
\smartqed  
\usepackage{graphicx}
\usepackage{pslatex}
\usepackage{amsfonts,color}
\usepackage{amssymb,amsmath,latexsym}
\usepackage[justification=centering]{caption}

\makeatletter
\def\inmod#1{\allowbreak\mkern5mu({\operator@font mod}\,\,#1)}

\newcommand{\Rmnum}[1]{\expandafter\@slowromancap\romannumeral #1@}
\makeatother

\begin{document}

\title{
A class of narrow-sense BCH codes
over $\mathbb{F}_q$ of length
$\frac{q^m-1}{2}$
}


\author{Xin Ling\and Sihem Mesnager \and Yanfeng Qi\and Chunming Tang }


\institute{
X. Ling \at
School of Mathematics and Information, China West Normal University, Nanchong
637002, China.
\email{xinlingcwnu@163.com}.
\and
S. Mesnager \at
              the Department of Mathematics, University of Paris
VIII, 93526 Saint-Denis, France, with LAGA UMR 7539,  CNRS, Sorbonne Paris
Cit\'e, University of Paris XIII, 93430 Paris,  France, and also with Telecom
ParisTech, 75013 Paris, France. \email{ smesnager@univ-paris8.fr}. \and
Y. Qi \at
                School of Science, Hangzhou Dianzi University, Hangzhou,
Zhejiang, 310018, China. 
\email{qiyanfeng07@163.com}.
\and
C. Tang \at
School of Mathematics and Information, China West Normal University, Nanchong
637002, China, and Department of Mathematics, The Hong Kong University of
Science and Technology, Clear Water Bay, Kowloon, Hong Kong, China.
              \email{tangchunmingmath@163.com}. 
}

\date{Received: date / Accepted: date}

\maketitle

\begin{abstract}
BCH codes with efficient encoding and decoding algorithms have many applications
in communications, cryptography and combinatorics design.  This paper studies a class of linear codes of length $
\frac{q^m-1}{2}$ over $\mathbb{F}_q$ with special trace representation, where $q$ is an odd prime power.  With the help of  the inner distributions of some subsets of association schemes from bilinear forms associated with quadratic forms, we determine the weight enumerators of these codes.  From determining
some  cyclotomic coset leaders $\delta_i$ of
cyclotomic cosets modulo  $
\frac{q^m-1}{2}$, we prove that
narrow-sense BCH codes of length $
\frac{q^m-1}{2}$ with designed distance
$\delta_i=\frac{q^m-q^{m-1}}{2}-1-\frac{q^{
\lfloor \frac{m-3}{2} \rfloor+i}-1}{2}$ have
the corresponding trace representation, and have the  minimal
distance $d=\delta_i$ and the Bose distance
$d_B=\delta_i$, where
$1\leq i\leq \lfloor \frac{m+3}{4} \rfloor$.

\end{abstract}

\section{Introduction}
As an important class of cyclic codes,
BCH codes with efficient encoding and decoding algorithms
have many applications
in communications, cryptography and combinatorics design.  They were independently discovered by  Hocquenghem \cite{H59} and
Bose, Ray-Chaudhuri \cite{BR60}, and were generalized from binary fields to finite fields by Gorenstein and Zierler \cite{GZ61}.  The determination of parameters of BCH codes is an interesting and difficult problem. There are many results on BCH codes
\cite{AKS07,B70,C90,C98,CHZ06,D15,D18,DDZ15,DFZ17,DFK97,FL08,FTKL86,HS73,K69,KT69,KL99,KL01,LDL-17,LLDL17,LDL17,MS77,M62,P69,W69,Y18,YF00,YH96,ZD14}.

Let $m$ be a positive integer and $\mathbb{F}_{q^m}$ be a finite field. Let
$\beta$ be an element in $\mathbb{F}_{q^m}$ with order $n$, where $n\mid (q^m-1)$.
A BCH code $\mathcal{C}$ over $\mathbb{F}_q$ with designed distance $\delta$ is a cyclic code with  a generator polynomial $g(x)$, where $g(x)$ is determined by its $\delta-1$ consecutive roots
$\beta^{l},\beta^{l+1},\ldots, \beta^{l+\delta-1}$. The minimal distance of
the BCH code $\mathcal{C}$ is greater than
$\delta$. When $n=q^m-1$, the code
$\mathcal{C}$ is called primitive. When
$l=1$, this code $\mathcal{C}$ is called narrow-sense. Two narrow-sense BCH codes with different designed distances may be the same. The Bose distance $d_B$ of a BCH code is the largest
designed distance. The Bose distance satisfies that $d_B\leq d\leq d_B+4$ \cite{C98}, where $d$ is the minimum distance. The determination of
the minimum distance $d$ and the Bose distance of a narrow-sense BCH
code attracts much interest
\cite{AKS07,C90,C98,DDZ15,KL72,KLP67,KT69,M80,M62,YF00}.
Some results on $[n,k,d]$ narrow-sense  BCH codes over $\mathbb{F}_{q}$ such that
$d=d_B$   are listed:
\begin{itemize}
\item the primitive case: $n=q^m-1$
\begin{itemize}
\item $d=d_B=q^i-1$, where
$1\leq i\leq m-1$ \cite{P69};
\item $q=2$, $d=d_B=2^{m-1}-2^i-1$, where
 $\frac{m-2}{2}\leq i\leq m-
\lfloor \frac{m}{3}\rfloor-1$  \cite{B70};
\item $q=2$, $d=d_B=2^{m-1-s}-2^{m-1-i-s}$, where  $1\leq i\leq m-s-2$ and
$0\leq s\leq m-2i$ \cite{KL72};
\item $d=d_B=q^i+1$, where $(m-i)\mid i$
or $2i\mid m$ \cite{DDZ15};
\item $d=d_B=(q-l_0)q^{m-l_1-1}-1$, where
$0\leq l_0\leq q-2$ and $0\leq l_1\leq m-1$  \cite{D15};
\item $d=d_B=q^m-q^{m-1}-q^i-1$, where
$\frac{m-2}{2}\leq i\leq m-\lfloor \frac{m}{3}\rfloor-1$  [Theorem 1.2, \cite{L17}].
\end{itemize}
\item $n=\frac{q^m-1}{2}$
\begin{itemize}
\item $q=3$, $d=d_B=
\frac{q^m-q^{m-1}}{2}-1-\frac{q^{
\lfloor \frac{m-3}{2} \rfloor+i}-1}{2}$ for
$i=1,2$ \cite{LDXG17};
\item $d=d_B=
\frac{q^m-q^{m-1}}{2}-1-\frac{q^{
\lfloor \frac{m-3}{2} \rfloor+i}-1}{2}$ for
$i=1,2$ \cite{ZSK19};
\item $d=d_B=
\frac{q^m-q^{m-1}}{2}-1-\frac{q^{
\lfloor \frac{m-3}{2} \rfloor+i}-1}{2}$, where $1\leq i\leq \lfloor \frac{m+3}{4}
\rfloor$  [this paper].
\end{itemize}
\end{itemize}

To determine the minimal distance $d$
and the Bose distance $d_B$ of
narrow-sense primitive BCH codes with designed
distance $q^m-q^{m-1}-q^i-1$, Li
\cite{L17}  employed the framework of
association schemes from bilinear forms
associated with quadratic forms and presented the weight enumerators of these BCH codes, where
these quadratic forms are corresponding to the trace representation of BCH codes.
In \cite{TDX19}, Tang et al. studied a class of
ternary linear codes with trace representation,
determined the weight distributions  of some shortened codes and punctured codes of these three-weight subcodes, and presented some
$2$-designs from these codes.

Motivated by these papers and their ideas, this paper studies the following  linear codes
$\mathcal{C}_{1,h}$ (resp. $\mathcal{C}_{2,h}$)  of length
$n=\frac{q^m-1}{2}$ with special trace representation over finite field
$\mathbb{F}_q$
for $m$ odd (resp. even), where $q$ is odd prime power and
$$
\mathcal{C}_{1,h}=\left\{
\left( \mathrm{Tr}_1^m\left (\sum_{j=h}^{\frac{m-1}{2}} a_j \alpha^{\left (q^{j}+1 \right )l} \right )+a \right )_{l=0}^{n-1}:
a_h,\ldots, a_{\frac{m-1}{2}}\in \mathbb{F}_{q^m}, a\in \mathbb{F}_{q} \right\}
$$
and
$$
\mathcal{C}_{2,h}=
\left\{ \left( \mathrm{Tr}_1^m\left (a_{\frac{m}{2}}\alpha^{\left (q^{\frac{m}{2}}+1 \right )l}+\sum_{j=h}^{\frac{m-2}{2}} a_j \alpha^{\left (q^{j}+1 \right )l} \right )+a \right )_{l=0}^{n-1}:
a_h,\ldots, a_{\frac{m-2}{2}}\in \mathbb{F}_{q^m}, a_{\frac{m}{2}}\in \mathbb{F}_{q^{\frac{m}{2}}},
a\in \mathbb{F}_q
 \right\}.
$$
These codes correspond to
quadratic forms. With the help of  the theory of
association schemes from bilinear forms associated with quadratic forms, we determine
the weight enumerators of these codes.
Determining some largest cyclotomic coset leaders of cyclotomic cosets modulo $n=\frac{q^m-1}{2}$,
we prove that a narrow-sense BCH code  with
designed distance $\delta_i=
\frac{q^m-q^{m-1}}{2}-1-\frac{q^{
\lfloor \frac{m-3}{2} \rfloor+i}-1}{2}$ has
the above trace representation and has  the minimal distance $d$ such that
$d=d_B=\delta_i$, where $1\leq i\leq \lfloor \frac{m+3}{4}
\rfloor$.

The rest of the paper is organized as follows.
Section 2 introduces some basic results on
BCH codes, quadratic forms and association schemes. Section 3 studies some linear codes with trace representation, uses association schemes to determine the weight enumerators of these codes, and presents the minimum distance $d$ and the Bose distance of the narrow-sense BCH codes of length $
n=\frac{q^m-1}{2}$ with designed distance
$\delta_i$. Section 4 makes a conclusion.
\section{Preliminaries}
In this section, some results on BCH codes,
quadratic forms and association schemes are introduced.

\subsection{BCH codes}
Let $q$ be a prime power. An $[n,k,d]$ linear code $\mathcal{C}$ over the finite field $\mathbb{F}_q$ is a
$k$-dimensional subspace of $\mathbb{F}_q^n$, where $d$ is the minimal distance of $\mathcal{C}$.
A linear code is called a cyclic code if
$(c_0,c_1,\ldots,c_{n-1})\in \mathcal{C}$ implies that  $(c_{n-1},c_0, \ldots,c_{n-2})\in \mathcal{C}$. We also write a cyclic code as a
principal ideal of the ring
$\mathbb{F}_q[x]/(x^n-1)$, since there is a map
\begin{align*}
\psi: \mathbb{F}_q^n&\longrightarrow \mathbb{F}_q[x]/(x^n-1)\\
(c_0,c_1,\ldots,c_{n-1})&
\longmapsto c_1+c_1x+\cdots+c_{n-1}x^{n-1}.
\end{align*}
Then $\mathcal{C}=\langle g(x)\rangle$, where
$g(x)$ is a monic polynomial and is called
the generator polynomial of $\mathcal{C}$.
Note that $g(x)|(x^n-1)$. The polynomial $
h(x)=g(x)|(x^n-1)$ is called the parity-check
polynomial of $\mathcal{C}$.
The  code $\mathcal{C}$ is called a
cyclic code with $s$ zeros (resp.
$s$ nonzeros) if $g(x)$
(resp. $h(x)$) can be factored into a product of $s$ irreducible polynomials over $\mathbb{F}_q$.

Let $n$ be a positive integer and $m$ be the smallest positive  integer such that
$n|(q^n-1)$. Let $\alpha$ be a primitive element of $\mathbb{F}_{q^m}$ and
$\beta=\alpha^{\frac{q^m-1}{n}}$. Let
$m_i(x)$ be the minimal polynomial of
$\beta^i$ over $\mathbb{F}_q$, where
$0\leq i\leq n-1$. For each
$2\leq \delta\leq n$, define two polynomials
$$
g_{(n,q,m,\delta)}(x)=lcm(m_1(x),\ldots,m_{
\delta-1}(x))  ~\text{and}~
\tilde{g}_{(n,q,m,\delta)}(x)= (x-1)
g_{(n,q,m,\delta)}(x),
$$
where $lcm$ denotes the least common multiple of the polynomials.
Let $\mathcal{C}_{(n,q,m,\delta)}
=\langle g_{(n,q,m,\delta)}(x) \rangle$ and
$\tilde{\mathcal{C}}_{(n,q,m,\delta)}
=\langle \tilde{g}_{(n,q,m,\delta)}(x) \rangle$. Then $\mathcal{C}_{(n,q,m,\delta)}$
is called a narrow-sense BCH code with designed distance $\delta$ and
$\tilde{\mathcal{C}}_{(n,q,m,\delta)}$ is the even-like subcode of $\mathcal{C}_{(n,q,m,\delta)}$. Note that
$dim(\tilde{\mathcal{C}}_{(n,q,m,\delta)})
=dim(\mathcal{C}_{(n,q,m,\delta)})-1$ and
the minimal distances of $\mathcal{C}_{(n,q,m,\delta)}$ and $\tilde{\mathcal{C}}_{(n,q,m,\delta)}$ are at least $\delta$ and $\delta+1$,
respectively.  If $n=q^m-1$, the code
$\mathcal{C}_{(n,q,m,\delta)}$ is called a narrow-sense primitive BCH code. Note that
two BCH codes with different designed distances
may be the same.
For a narrow-sense BCH code
$\mathcal{C}_{(n,q,m,\delta)}$,
the Bose distance \cite{MS77} of this code is defined by the largest designed distance
$d_B=max (\delta')$, where
${\mathcal{C}_{(n,q,m,\delta')}=
\mathcal{C}_{(n,q,m,\delta)}}$.

The narrow-sense  BCH code $\mathcal{C}_{(n,q,m,\delta)}$ has close relation with $q$-cyclotomic cosets modulo $n$.  The $q$-cyclotomic coset of $s$ modulo $n$ is defined by
$$
C_s=\{s,sq,sq^2,\ldots,sq^{l_s-1}\}
\mod n \subseteq \mathbb{Z}_n,
$$
where $0\leq s\leq n-1$ and
$l_s$ called the size of the
$q$-cyclotomic coset  is the smallest positive integer such that $s\equiv sq^{l_s} \mod n$.
Note that  $l_s\mid m$.
The coset leader of $C_s$ is the smallest integer in $C_s$. Let
$\Gamma_{n,q}$ be the set of
all the coset leaders. Then $C_s\bigcap C_t
=\emptyset$ and $\mathbb{Z}_n=\bigcup_{s
\in \Gamma_{(n,q)}}C_s$, where $s,t
\in \Gamma_{(n,q)}$ and $s\neq t$.

Let $m_s(x)$ be the minimal polynomial of
$\beta^s$ over $\mathbb{F}_q$. Then
$$
m_s(x)=\prod_{i\in C_s}(x-\beta^i)
\in \mathbb{F}_q[x],
$$
$deg(m_s(x))=l_s$ and
$$
x^n-1
=\prod_{s\in \Gamma_{(n,q)}}
m_{s}(x).
$$
Hence, the generator polynomial
$g_{(n,q,m,\delta)}(x)=
\prod_{
s\in \cup_{j=1}^{\delta-1}C_j}m_s(x)$,
the parity-check
polynomial of $\mathcal{C}_{(n,q,m,\delta)}$ is $$
h(x)=(x-1)\prod_{s\geq \delta,
s\in \Gamma_{(n,q)}}m_s(x),
$$
and the dimension of  $\mathcal{C}_{(n,q,m,\delta)}$ is
$$
\sum_{s\geq \delta, s\in \Gamma_{(n,q)}}l_s+1.
$$
The following lemma gives the relation between the Bose distance $d_B$ of $\mathcal{C}_{(n,q,m,\delta)}$ and coset leaders.
\begin{lemma}[Proposition 4, \cite{LDXG17}]\label{bose-d}
The Bose distance $d_B$ of the narrow-sense BCH code $\mathcal{C}_{(n,q,m,\delta)}$ is a coset leader of a $q$-cyclotomic coset modulo $n$.
Furthermore, if $\delta$ is a coset leader, then $d_B=\delta$.
\end{lemma}

From Delsarte's theorem \cite{D75}, a cyclic code has the following  trace representation.
\begin{proposition}
Let $q$ be a prime power, $n$ be a positive integer, $m$ be the smallest positive integer such that $n|(q^n-1)$, and
$\beta$ be  an $n$-th primitive element of
$\mathbb{F}_{q^m}$.   Let
$\mathcal{C}$ be a cyclic code of
length $n$ over $\mathbb{F}_{q}$ with
$s$ nonzeros. Let $\beta^{i_1},
\ldots, \beta^{i_s}$ be $s$ roots of
its parity-check polynomial $h(x)$, which are not conjugate with each other. Then
the code $\mathcal{C}$ has the following trace representation
$$
\mathcal{C}=
\{c(a_1,\ldots,a_s): a_j\in \mathbb{F}_{q^{l_j}}, ~\text{for}~1\leq j\leq s\},
$$
where $l_j$ is the size of
the $q$-cyclotomic coset $C_{i_j}$, $\mathrm{Tr}_1^{l_j}$ is the trace function
from $\mathbb{F}_{q^{l_j}}$ to
$\mathbb{F}_q$, and $c(a_1,\ldots,a_s)=
\left( \sum_{j=1}^{s}
\mathrm{Tr}_1^{l_j}(a_j\beta^{-li_j}) \right)_{l=0}^{n-1}$.
\end{proposition}

\subsection{Quadratic forms}
This subsection introduces some results on
quadratic forms \cite{LN97}.
Let $q$ be an odd prime power and
$V$ be an $n$-dimensional vector space over
$\mathbb{F}_q$.
A quadratic form $Q$ on $V$   is a function from $V$ to $\mathbb{F}_q$ satisfying
$Q(\lambda x)=\lambda^2Q(x)$, where
$x\in V$ and $\lambda\in \mathbb{F}_q$.
A symmetric bilinear form $B_Q(x,y)
=\frac{1}{2}(Q(x+y)-Q(x)-Q(y))$
on $V$ is associated with $Q$. The radical of $Q$ is
$Rad(Q)=Q^{-1}(0)\cap Rad(B_Q)$, where
$Rad(B_Q)=\{y\in V: B(x,y)=0,
\forall x\in V\}$
is the radical of the symmetric bilinear form
$B_Q$. The radical of $Q$ is a vector space over $\mathbb{F}_q$. The rank of
$Q$ is $rank(Q)=n-dim(Rad(Q))$ and
the rank of $B_Q$ is $rank(B_Q)=n-dim(Rad(B_Q))$.
\begin{lemma}[Lemma 3.6, \cite{L17}]\label{rankQB}
Let $q$ be an odd prime power, $Q$ be a quadratic form on the $n$-dimensional vector space $V$ over $\mathbb{F}_q$ and
$B_Q$ be its associated bilinear form.
Then $rank(Q)=rank(B_Q)$.
\end{lemma}

Let $V=\mathbb{F}_q^n$ and $Q$ be a quadratic form on $V$. Then $Q$  has the following expression $Q(x)=\sum_{1\leq i,j\leq n}
c_{ij}x_ix_j$. Two quadratic forms
$Q$ and $Q'$ are equivalent if there is an
$n\times n$ nonsingular matrix $A$ such that
$Q(x_1,\ldots,x_n)=Q'((x_1,\ldots,x_n)A)$.
Every quadratic form $Q$ is equivalent to
$\sum_{i=1}^ra_ix_i^2$, where
$a_i\in \mathbb{F}_q^*$ and $r=rank(Q)$.
The type of
$Q$ is $\tau=\eta(\prod_{i=1}^ra_i)$, where
$\eta$ is the quadratic character of
$\mathbb{F}_q$.
\begin{lemma}[Lemma 5.1, \cite{S15}]\label{NQb}
Let $q$ be an odd prime power and
$Q$ be quadratic form of rank $r$ and type
$\tau$. Let $N(b)$ be the number of solutions
$Q(x)=b$, where $h\in \mathbb{F}_q$. Then
$$
N(b)=\left\{
       \begin{array}{ll}
         q^{n-1}+\tau \eta(-1)^{\frac{r-1}{2}}
\eta(b)q^{n-\frac{r+1}{2}}, & \hbox{if $r$ is odd;} \\
         q^{n-1}+\tau \eta(-1)^{\frac{r}{2}}
v(b)q^{n-\frac{r+2}{2}}, & \hbox{if $r$ is even.}
       \end{array}
     \right.
$$
where $\eta(0)=0$, $v(x)=-1$ for $x\in \mathbb{F}_q^*$ and
$v(0)=q-1$.
\end{lemma}

Let $V=\mathbb{F}_{q^n}$.  Choose a basis
$\alpha_1,\ldots,\alpha_n$ of $V$ over
$\mathbb{F}_q$. We have a bijection between
$\mathbb{F}_{q^n}$ and $\mathbb{F}_q^n$. Note that the rank $r$ and the type $\tau$ of a quadratic form $Q$ on $V$ are independent of the choice of the basis.

\subsection{Association schemes}
This subsection introduces some results on association schemes
\cite{MS77,HW93,WWMM03}.
Let $X$ be a finite set and
$R_0,R_1,\ldots, R_n$ be a partition of
$X\times X$. Then a pair $\left(X,
(R_i)\right)$ is called an association scheme
with $n$ classes if it satisfies
\begin{itemize}
\item $R_0=\{(x,x): x\in X\}$;
\item for each $i$, there exists $j$ such that the inverse of $R_i$ equals $R_j$;
\item if $(x,y)\in R_k$, the number of the set
$\{z\in X: (x,z)\in R_i,
(z,y)\in R_j\}$ is a constant $p_{ij}^k$ depending on only $i$, $j$, and $k$, but not on the particular choice of $x$ and $y$.
\end{itemize}
An association scheme $\left(X,(R_i)\right)$ is symmetric if for all $i$ the inverse of $R_i$ equals $R_i$. It is commutative if
for all $i$, $j$, and $k$, $p_{ij}^k
=p_{ji}^k$.

Let $\left(X,(R_i)\right)$ be a commutative
association scheme with $n$ classes and
$A_i$ be the adjacency matrix of the diagraph
$(X,R_i)$. The matrices $A_0,A_1,\ldots,
A_n$ span a vector space  over the complex numbers called the Bose-Mesner algebra of
$\left(X,(R_i)\right)$ with dimension $n+1$.
This vector space   has another uniquely defined basis consisting of minimal idempotent
matrices $E_0,E_1,\ldots, E_n$. Then
$$
A_i=\sum_{k=0}^nP_i(k)E_k
~\text{and}~ E_k=\frac{1}{|X|}
\sum_{i=0}^nQ_k(i)A_i.
$$
The uniquely defined numbers $P_i(k)$ and
$Q_k(i)$ are called the $P$-numbers and
the $Q$-numbers of $\left(X,(R_i)\right)$,
respectively.

Let  $V$ be a vector space over the finite field
$\mathbb{F}_q$ of dimension $m$ and
$X(m,q)$ be the set of symmetric bilinear forms on $V$, where
$q$ is an odd prime power.
Let $\alpha_1,\ldots, \alpha_m$ be a basis of $V$. Then a symmetric bilinear form $B\in
X(m,q)$ has the $m\times m$
 symmetric matrix
$$
\left(B(\alpha_i,\alpha_j) \right)_{1\leq
i,j\leq m}.
$$
The rank of $B$ is the rank of this matrix, which is independent of the choice of the basis. This matrix is congruent to
a diagonal matrix, whose diagonal is either
zero of $[z,1,\ldots,1,0,\ldots,0]$ for
some nonzero $z\in \mathbb{F}_q$. The type of
$B$ is $\eta(z)$, where $\eta$ is the quadratic character of $\mathbb{F}_q$.

Let $X_{r,\tau}$ be the set of all the symmetric bilinear forms with rank $r$ and
type $\tau$. Define
$$
R_{r,\tau}=
\{(A,B)\in X(m,q)\times X(m,q):
A-B\in X_{r,\tau}\}.
$$
Then $\left( X(m,q), (R_{r,\tau}) \right)$ is an association scheme with $2m$ classes.

Let $Y$ be a subset of $X(m,q)$ and
the inner distribution of $Y$ be the sequence of numbers  $(a_{r,\tau})$, where
$$
a_{r,\tau}=\frac{|(Y\times Y)\cap R_{r,\tau}|}{|Y|}.
$$
Let $Q_{k,\epsilon}(r,\tau)$ be the
$Q$-numbers of $\left(X(m,q),
(R_{r,\tau}) \right)$. The dual inner distribution of $Y$ is the sequence of
$(a_{k,\epsilon}')$, where
$$
a_{k,\epsilon}'=\sum_{r,\tau}
Q_{k,\epsilon}(r,\tau)a_{r,\tau}.
$$
Note that $\tau\in \{1,-1\}$. When
$r=0$, we write $a_0$.
\begin{definition}
Let $Y$ be a subset of $X(m,q)$. The set $Y$ is a $d$-code if
$a_{i,1}=a_{i,-1}=0$ for each
$i\in \{1,2,\ldots,d-1\}$. The set
$Y$ is a proper $d$-code if
it is a $d$-code and it is not a $(d+1)$-code.
The set $Y$ is a $t$-design if
$a_{i,1}'=a_{i,-1}'=0$ for each
$i\in \{1,2,\ldots,t\}$.
The set $Y$ is a $(2t+1,\epsilon)$-design if
it is a $(2t+1)$-design and
$a_{2t+2,\epsilon}'=0$.
\end{definition}
Note that the designs involved in  associate schemes are not the usual t-designs
studied in combinatorial design theory.
Define $q^2$-analogs of binomial coefficients
$$
\begin{bmatrix}
n\\
k
\end{bmatrix}
=\prod_{i=1}^{k}\frac{q^{2n-2i+2}-1}
{q^{2i}-1}
$$
for integers $n$ and $k\geq 0$. Note that
$\begin{bmatrix}
n\\
0
\end{bmatrix}=1
$.  The inner distribution of $Y$ can be given by $q^2$-analogs of binomial coefficients in the following theorem and proposition.
\begin{theorem}[Theorem 3.9, \cite{S15}]\label{tart}
If $Y$ is a $(2\delta-1)$-code and
a $(2n-2\delta+3)$-design in $X(2n+1,q)$, then the inner distribution $(a_{r,\tau})$ of
$Y$ satisfies
\begin{align*}
a_{2i-1,\tau}&=\frac{1}{2}
\begin{bmatrix}
n\\
i-1
\end{bmatrix}
\sum_{j=0}^{i-\delta}(-1)^jq^{j(j-1)}
\begin{bmatrix}
i\\
j
\end{bmatrix}
(\frac{|Y|}{q^{(2n+1)(n+1+j-i)}}-1)\\
a_{2i,\tau}&=\frac{1}{2}(q^{2i}+
\tau\eta(-1)^iq^i)
\begin{bmatrix}
n\\
i
\end{bmatrix}
\sum_{j=0}^{i-\delta}(-1)^jq^{j(j-1)}
\begin{bmatrix}
i\\
j
\end{bmatrix}
(\frac{|Y|}{q^{(2n+1)(n+1+j-i)}}-1)
\end{align*}
for $i>0$.
If $Y$ is a $(2\delta-1)$-code and
a $(2n-2\delta+2)$-design in $X(2n,q)$, then the inner distribution $(a_{r,\tau})$ of
$Y$ satisfies
\begin{align*}
a_{2i-1,\tau}=&\frac{1}{2}(q^{2i}-1)
\begin{bmatrix}
n\\
i
\end{bmatrix}
\sum_{j=0}^{i-\delta}(-1)^jq^{j(j-1)}
\begin{bmatrix}
i-1\\
j
\end{bmatrix}
\frac{|Y|q^{2j}}{q^{(2n+1)(n+1+j-i)}}\\
a_{2i,\tau}=&\frac{1}{2}
\begin{bmatrix}
n\\
i
\end{bmatrix}
\sum_{j=0}^{i-\delta+1}(-1)^jq^{j(j-1)}
\begin{bmatrix}
i\\
j
\end{bmatrix}
(\frac{|Y|q^{2j}}{q^{(2n+1)(n+j-i)}}-1)\\
&+\frac{\tau}{2}\eta(-1)^iq^i
\begin{bmatrix}
n\\
i
\end{bmatrix}
\sum_{j=0}^{i-\delta}(-1)^jq^{j(j-1)}
\begin{bmatrix}
i\\
j
\end{bmatrix}
(\frac{|Y|}{q^{(2n-1)(n+j-i)}q^{2n}}-1)
\end{align*}
for $i>0$.
\end{theorem}

\begin{proposition}[Proposition 3.10, \cite{S15}]\label{part}
If $Y$ is a $(2\delta)$-code and
a $(2n-2\delta+1)$-design in $X(2n,q)$, then the inner distribution $(a_{r,\tau})$ of
$Y$ satisfies
\begin{align*}
a_{2i-1,\tau}=&\frac{1}{2}(q^{2i}-1)
\begin{bmatrix}
n\\
i
\end{bmatrix}
\sum_{j=0}^{i-\delta-1}(-1)^jq^{j(j-1)}
\begin{bmatrix}
i-1\\
j
\end{bmatrix}
\frac{|Y|q^{2j}}{q^{(2n+1)(n+1+j-i)}}\\
a_{2i,\tau}=&\frac{1}{2}
\begin{bmatrix}
n\\
i
\end{bmatrix}
\sum_{j=0}^{i-\delta}(-1)^jq^{j(j-1)}
\begin{bmatrix}
i\\
j
\end{bmatrix}
(\frac{|Y|q^{2j}}{q^{(2n+1)(n+j-i)}}-1)\\
&+ \frac{\tau}{2}\eta(-1)^iq^i
\begin{bmatrix}
n\\
i
\end{bmatrix}
\sum_{j=0}^{i-\delta}(-1)^jq^{j(j-1)}
\begin{bmatrix}
i\\
j
\end{bmatrix}
(\frac{|Y|}{q^{(2n-1)(n+j-i)}q^{2n}}-1)
\end{align*}
for $i>0$.
If $Y$ is a $(2\delta)$-code and
a $(2n-2\delta+1,\eta(-1)^{n-\delta+1})$-design in $X(2n+1,q)$, then the inner distribution $(a_{r,\tau})$ of
$Y$ satisfies
\begin{align*}
a_{2i-1,\tau}=&\frac{1}{2}
\begin{bmatrix}
n\\
i-1
\end{bmatrix}
\sum_{j=0}^{i-\delta}(-1)^jq^{j(j-1)}
\begin{bmatrix}
i\\
j
\end{bmatrix}
(\frac{|Y|}{q^{(2n+1)(n+1+j-i)}}-1)\\
&+ \frac{1}{2}(-1)^{i-\delta}q^{(i-\delta)
(i-\delta-1)}
\begin{bmatrix}
n\\
\delta-1
\end{bmatrix} (\frac{|Y|}{q^{(2n+1)(n-\delta+1)}}-1)
\left(
\begin{bmatrix}
n-\delta\\
n-i+1
\end{bmatrix}
(q^{n-\delta+1}+1)-
\begin{bmatrix}
n-\delta+1\\
n-i+1
\end{bmatrix}
\right)\\
a_{2i,\tau}=&\frac{1}{2}(q^{2i}+
\tau\eta(-1)^{i}q^i)
\begin{bmatrix}
n\\
i
\end{bmatrix}
\sum_{j=0}^{i-\delta}(-1)^jq^{j(j-1)}
\begin{bmatrix}
i\\
j
\end{bmatrix}
(\frac{|Y|}{q^{(2n+1)(n+1+j-i)}}-1)\\
&+\frac{1}{2}\eta(-1)^{i-\delta}q^{(i-\delta+1)
(i-\delta)}
\begin{bmatrix}
n\\
\delta-1
\end{bmatrix}
\begin{bmatrix}
n-\delta\\
n-i
\end{bmatrix}
(q^{n-\delta+1}+1)
(\frac{|Y|}{q^{(2n+1)(n-\delta+1)}}-1)
\end{align*}
for $i>0$.
\end{proposition}

\section{A class of linear codes of 
length $n=\frac{q^m-1}{2}$}
In this section, let $q$ be an odd prime
power, $m$ be a positive integer, $\alpha$ be a generator of $\mathbb{F}_{q^m}^*$ and $\beta=\alpha^{2}$. Then $\beta$ is a primitive $n$-th root of unity
in $\mathbb{F}_{q^m}$,
where $n=\frac{q^m-1}{2}$.

Let $h$ be a positive integer.
When $m$ is odd, define the codes
$$
\mathcal{C}_{1,h}=\left\{
\left( \mathrm{Tr}_1^m\left (\sum_{j=h}^{\frac{m-1}{2}} a_j \alpha^{\left (q^{j}+1 \right )l} \right )+a \right )_{l=0}^{n-1}:
a_h,\ldots, a_{\frac{m-1}{2}}\in \mathbb{F}_{q^m}, a\in \mathbb{F}_{q} \right\}
$$
and
$$
\tilde{\mathcal{C}}_{1,h}=\left\{
\left( \mathrm{Tr}_1^m\left (\sum_{j=h}^{\frac{m-1}{2}} a_j \alpha^{\left (q^{j}+1 \right )l} \right ) \right )_{l=0}^{n-1}:
a_h,\ldots, a_{\frac{m-1}{2}}\in \mathbb{F}_{q^m}, a\in \mathbb{F}_{q} \right\}.
$$
When $m$ is even,  define the codes
$$
\mathcal{C}_{2,h}=
\left\{ \left( \mathrm{Tr}_1^m\left (a_{\frac{m}{2}}\alpha^{\left (q^{\frac{m}{2}}+1 \right )l}+\sum_{j=h}^{\frac{m-2}{2}} a_j \alpha^{\left (q^{j}+1 \right )l} \right )+a \right )_{l=0}^{n-1}:
a_h,\ldots, a_{\frac{m-2}{2}}\in \mathbb{F}_{q^m}, a_{\frac{m}{2}}\in \mathbb{F}_{q^{\frac{m}{2}}},
a\in \mathbb{F}_q
 \right\}
$$
and
$$
\tilde{\mathcal{C}}_{2,h}=
\left\{ \left( \mathrm{Tr}_1^m\left (a_{\frac{m}{2}}\alpha^{\left (q^{\frac{m}{2}}+1 \right )l}+\sum_{j=h}^{\frac{m-2}{2}} a_j \alpha^{\left (q^{j}+1 \right )l} \right ) \right )_{l=0}^{n-1}:
a_h,\ldots, a_{\frac{m-2}{2}}\in \mathbb{F}_{q^m}, a_{\frac{m}{2}}\in \mathbb{F}_{q^{\frac{m}{2}}},
a\in \mathbb{F}_q
 \right\}.
$$

Let $Q(x)$ be a quadratic form defined by
$$
Q(x)=\mathrm{Tr}_1^m\left(
\sum_{j=1}^la_jx^{q^{k_j}+1}
\right), 
$$
where $k_j\geq0$ and $a_j\in \mathbb{F}_{q^m}$. The symmetric bilinear form associated with
$Q(x)$ is
\begin{align*}
B_Q(x,y)&=\frac{1}{2}\left( Q(x+y)-Q(x)-Q(y)
\right)\\
&= \frac{1}{2} \sum_{j=0}^l
\mathrm{Tr}_1^m\left(
a_j(x^{q^{k_j}}+y^{q^{k_j}})(x+y)
-a_jx^{q^{k_j}+1}-a_jy^{q^{k_j}+1}
\right)\\
&=\frac{1}{2}\sum_{j=1}^{l}\mathrm{Tr}_1^m
(a_jx^{q^{k_j}}y+a_jxy^{q^{k_j}})\\
&=\frac{1}{2}\sum_{j=1}^{l}\mathrm{Tr}_1^m
(a_jx^{q^{k_j}}y)+
\frac{1}{2}\sum_{j=1}^{l}\mathrm{Tr}_1^m
(a_jxy^{q^{k_j}})\\
&=\sum_{j=1}^{l}\mathrm{Tr}_1^m
\left((\frac{a_j}{2})x^{q^{k_j}}y\right)+
\sum_{j=1}^{l}\mathrm{Tr}_1^m
\left((\frac{a_j}{2})^{q^{-k_j}}
x^{q^{-k_j}}y\right)\\
&=\sum_{j=0}^{l}
\mathrm{Tr}_1^m\left(
\left((\frac{a_j}{2})x^{q^{k_j}}+
(\frac{a_j}{2})^{q^{-k_j}}x^{q^{-k_j}}\right)y
\right).
\end{align*}

Let $h\geq 1$. Define the  set of quadratic forms:
\begin{align*}
Q_1&=\left\{
\mathrm{Tr}_1^m\left(\sum_{j=h}^{\frac{m-1}{2}}
 a_jx^{q^{j}+1}
\right): a_j\in\mathbb{F}_{q^m}
~\text{for}~ h\leq j\leq \frac{m-1}{2}
\right\},
\end{align*}
where $m$ is odd.  Define the set of quadratic forms:
\begin{align*}
Q_2&=\left\{
\mathrm{Tr}_1^m\left(
a_{\frac{m}{2}}x^{q^{\frac{m}{2}+1}}+
\sum_{j=h}^{\frac{m-2}{2}}a_jx^{q^{j}+1}
\right): a_{\frac{m}{2}} \in\mathbb{F}_{q^{
\frac{m}{2}}}
, a_j\in\mathbb{F}_{q^m}
~\text{for}~ h\leq j\leq \frac{m-2}{2}
\right\},
\end{align*}
where $m$ is even.
Then we have the following two sets of bilinear forms associated with $Q_1$ and $Q_2$ respectively.
\begin{align*}
S_1&=\left\{
\mathrm{Tr}_1^m\left(
\sum_{j=h}^{\frac{m-1}{2}}\left((\frac{a_j}{2})x^{q^{j}}+
(\frac{a_j}{2})^{q^{-j}}x^{q^{-j}}\right)y
\right):  a_j\in\mathbb{F}_{q^m}
~\text{for}~ h\leq j\leq \frac{m-1}{2}
\right\}\\
&=\left\{
\mathrm{Tr}_1^m\left(\sum_{j=h}^{\frac{m-1}{2}}
\left(a_jx^{q^{j}}+
a_j^{q^{-j}}x^{q^{-j}}\right)y
\right): a_j\in\mathbb{F}_{q^m}
~\text{for}~ h\leq j\leq \frac{m-1}{2}
\right\}
\end{align*}
and
\begin{align*}
S_2&=\left\{
\mathrm{Tr}_1^m\left(
\frac{a_{\frac{m}{2}}}{2}x^{\frac{m}{2}}+
\sum_{j=h}^{\frac{m-2}{2}}\left((\frac{a_j}{2})x^{q^{j}}+
(\frac{a_j}{2})^{q^{-j}}x^{q^{-j}}\right)y
\right): a_{\frac{m}{2}} \in\mathbb{F}_{q^{
\frac{m}{2}}}
, a_j\in\mathbb{F}_{q^m}
~\text{for}~ h\leq j\leq \frac{m-2}{2}
\right\}\\
&=\left\{
\mathrm{Tr}_1^m\left(
{a_{\frac{m}{2}}}x^{\frac{m}{2}}+
\sum_{j=h}^{\frac{m-2}{2}}
\left(a_jx^{q^{j}}+
a_j^{q^{-j}}x^{q^{-j}}\right)y
\right): a_{\frac{m}{2}} \in\mathbb{F}_{q^{
\frac{m}{2}}}
, a_j\in\mathbb{F}_{q^m}
~\text{for}~ h\leq j\leq \frac{m-2}{2}
\right\}.
\end{align*}
Then $|S_1|=|S_2|=q^{m(\frac{m+1}{2}-h)}$.
The inner distributions of $S_1$ and
$S_2$ are given in the following proposition.
\begin{proposition}\label{code-design}
Let $(a_{0},a_{1,1},a_{1,-1},\ldots, a_{m,1},a_{m,-1})$ be the inner distribution of
$S_i$. If $m$ is odd, then $S_1$ is a proper $(2h+1)$-code and
$(m+1-2h)$-design in $X(m,q)$, and
the inner distribution of
$S_1$ is given in
Theorem \ref{tart}.
If $m$ is even, then $S_2$ is a proper $(2h)$-code and
$(m+1-2h)$-design in $X(m,q)$, and
the inner distribution of
$S_2$ is given in
Proposition \ref{part}.
\end{proposition}
\begin{proof}
We first prove that $S_1$ is a $(m+1-2h)$-design in $X(m,q)$, where $m$ is odd.

Let $U$ be a $t$-dimensional subspace of $\mathbb{F}_{q^m}$ and $A$ be a symmetric
bilinear form on $U$, where $t=m+1-2h$.
For $a=
(a_0,a_1,\ldots,m_{t-1})\in \mathbb{F}_{q^m}^t$, define a bilinear
form $B_a$ on $\mathbb{F}_{q^m}^t$
$$
B_a(x,y)=\mathrm{Tr}_1^m\left(\sum_{j=0}^{t-1}
a_jx^{q^{s(j+h})}y\right)
=\mathrm{Tr}_1^m\left(
\sum_{j=h}^{m-h}
a_{j-h}x^{q^{sj}}y\right).
$$
From Lemma 4.6 in \cite{S15}, the set
$\{B_a\mid_U: a \in \mathbb{F}_{q^m}^t\}$
is a multiset in which each
bilinear form on $U$ occurs a constant number  (depending only on $t$) of times.
From Lemma 4.3 in \cite{L17},
the number of elements in the multiset
$\left\{D+D': D\in \{B_a: a \in
\mathbb{F}_{q^m}^t\}\right\}$ that are an extension of $A$ is a constant independent of $U$ and $A$.

For $a=
(a_0,a_1,\ldots,m_{t-1})\in \mathbb{F}_{q^m}^t$, we have
\begin{align*}
B_a(x,y)+B_{a}'(x,y)
=& \mathrm{Tr}_1^m\left(
\sum_{j=h}^{m-h}
a_{j-h}(x^{q^{sj}}y+
xy^{q^{sj}})\right) \\
=&  \mathrm{Tr}_1^m\left(
\sum_{j=h}^{\frac{m-1}{2}}
a_{j-h}(x^{q^{sj}}y+
xy^{q^{sj}})\right)+\mathrm{Tr}_1^m\left(
\sum_{j=\frac{m+1}{2}}^{m-h}
a_{j-h}(x^{q^{sj}}y+
xy^{q^{sj}})\right)\\
=&
\mathrm{Tr}_1^m\left(
\sum_{j=h}^{\frac{m-1}{2}}
a_{j-h}(x^{q^{sj}}y+
xy^{q^{sj}})\right)+\mathrm{Tr}_1^m\left(
\sum_{j=h}^{\frac{m-1}{2}}
a_{m-j-h}(x^{q^{s(m-j)}}y+
xy^{q^{s(m-j)}})\right)\\
=&
\mathrm{Tr}_1^m\left(
\sum_{j=h}^{\frac{m-1}{2}}
a_{j-h}(x^{q^{sj}}y+
xy^{q^{sj}})\right)+\mathrm{Tr}_1^m\left(
\sum_{j=h}^{\frac{m-1}{2}}
a_{m-j-h}^{sj}(x^{q^{sm}}y^{q^{sj}}+
x^{q^{sj}}y^{q^{sm}})\right)\\
=&
\mathrm{Tr}_1^m\left(
\sum_{j=h}^{\frac{m-1}{2}}
a_{j-h}(x^{q^{sj}}y+
xy^{q^{sj}})\right)+\mathrm{Tr}_1^m\left(
\sum_{j=h}^{\frac{m-1}{2}}
a_{m-j-h}^{sj}(xy^{q^{sj}}+
x^{q^{sj}}y)\right)\\
=&
\mathrm{Tr}_1^m\left(
\sum_{j=h}^{\frac{m-1}{2}}
(a_{j-h}+a_{m-j-h}^{sj})(x^{q^{sj}}y+
xy^{q^{sj}})\right).
\end{align*}
Then
$$
\left\{D+D': D\in \{B_a: a \in
\mathbb{F}_{q^m}^t\}\right\}
=\left\{\mathrm{Tr}_1^m\left(
\sum_{j=h}^{\frac{m-1}{2}}
(a_{j-h}+a_{m-j-h}^{sj})(x^{q^{sj}}y+
xy^{q^{sj}})\right): a \in
\mathbb{F}_{q^m}^t\right\}.
$$
Note that $\{a_{j-h}, a_{m-j-h}: h\leq j\leq \frac{m-1}{2}\}=\{a_j: 0\leq j \leq t-1\}$.
Then $a_{j-h}+a_{m-j-h}^{sj}$ ranges over
$\mathbb{F}_{q^m}$ for $q^m$ times when
$a_{j-h}$ and $a_{m-j-h}$ ranges over
$\mathbb{F}_{q^m}$. Hence, the number of
elements in $S_1$ which are an extension of
$A$ is a constant independent of $U$ and $A$.
From Theorem 3.11 in \cite{S15}, the set
$S_1$ is a $t$-design in
$X(m,q)$.

We then prove that
$S_1$ is a proper $(2h+1)$-code.  For a bilinear form
$$
B(x,y)=
\mathrm{Tr}_1^m\left(\sum_{j=h}^{\frac{m-1}{2}}
\left(a_jx^{q^{j}}+
a_j^{q^{-j}}x^{q^{-j}}\right)y
\right)\in S_1,
$$
 we have
\begin{align*}
Rad(B(x,y))=&
\left\{ x\in \mathbb{F}_{q^m} :
\mathrm{Tr}_1^m\left(\sum_{j=h}^{\frac{m-1}{2}}
\left(a_jx^{q^{j}}+
a_j^{q^{-j}}x^{q^{-j}}\right)y
\right)=0, \forall y \in \mathbb{F}_{q^m}
\right\}\\
=& \left\{ x\in \mathbb{F}_{q^m} :
\mathrm{Tr}_1^m\left(\sum_{j=h}^{\frac{m-1}{2}}
\left(a_jx^{q^{j}}+
a_j^{q^{m-j}}x^{q^{m-j}}\right)y
\right)=0, \forall y \in \mathbb{F}_{q^m}
\right\}\\
=& \left\{ x\in \mathbb{F}_{q^m} :
 \sum_{j=h}^{\frac{m-1}{2}}
\left(a_jx^{q^{j}}+
a_j^{q^{m-j}}x^{q^{m-j}}\right)y
 =0, \forall y \in \mathbb{F}_{q^m}
\right\}\\
=& \left\{ x\in \mathbb{F}_{q^m} :
 \sum_{j=h}^{\frac{m-1}{2}}
\left(a_j^{-h}x^{q^{j-h}}+
a_j^{q^{m-j-h}}x^{q^{m-j-h}}\right)y
 =0, \forall y \in \mathbb{F}_{q^m}
\right\}.
\end{align*}
Since the degree of the linearized polynomial
$\sum_{j=h}^{\frac{m-1}{2}}
\left(a_j^{-h}x^{q^{j-h}}+
a_j^{q^{m-j-h}}x^{q^{m-j-h}}\right)y$
over $\mathbb{F}_{q^m}$ is at most
$m-2h$, the dimension of the vector space  $Rad(B(x,y))$ is at most
$m-2h$. Hence, $rank(B(x,y))=m
-dim(Rad(B(x,y)))\geq 2h$. Hence,
$S_1$ is a $(2h)$-code in $X(m,q)$.
Since $S_1$ is a
$(2m+1-2h)$-design, then
$S_1$ is a $(2m-2h, \tau)$-design.
From Proposition \ref{part} with $\delta=h$, we have
$a_{2h,\tau}=0$ for $\tau=\{1,-1\}$. Hence,
$S_1$ is a $(2h+1)$-code. From
Theorem \ref{tart}, we have
$a_{2h+1,\tau}>0$. Hence, $S_1$ is a proper
$(2h+1)$-code.

From a similar discussion, we have the corresponding results for $S_2$ when $m$ is even.
\end{proof}

\begin{lemma}\label{qurt-a}
Let $q$ be odd and $Q$ be a quadratic form of
rank $r\geq 1$ and type $\tau$ on
$\mathbb{F}_{q^m}$. Then the weight enumerator
$W_{r,\tau}$ of $\left(Q(\alpha^l)+a \right)_{l=0}^{n-1}$   is
\begin{align*}
Z^{\frac{q^m-q^{m-1}}{2}}+\frac{q-1}{2}
Z^{\frac{1}{2}\left( q^m-q^{m-1}-\tau\eta(-1)^{\frac{r-1}{2}}
q^{m-\frac{r+1}{2}}-1 \right)}
+\frac{q-1}{2}Z^{\frac{1}{2}\left( q^m-q^{m-1}+\tau\eta(-1)^{\frac{r-1}{2}}
q^{m-\frac{r+1}{2}}-1 \right)}
\end{align*}
if $r$ is odd, or
$$
Z^{\frac{q^m-q^{m-1}  - \tau \eta(-1)^{\frac{r}{2}} (q-1) q^{m-\frac{r+2}{2}}} {2}}+(q-1)
Z^{\frac{1}{2}\left( q^m-q^{m-1}+\tau\eta(-1)^{\frac{r}{2}}
q^{m-\frac{r+2}{2}}-1 \right)}
$$
if $r$ is even.
\end{lemma}
\begin{proof}
Let $f(x)=Q(x)+a$. We just computer the weight
$wt(f)$. Let $N(f=0)$ be the number of solutions of $f(x)=0$.

When $a=0$, from Lemma \ref{NQb}, we have
\begin{align*}
N(f=0)=
\left\{
  \begin{array}{ll}
    q^{m-1}, & \text{ if } r \text{ is odd};\\
    q^{m-1}  + \tau \eta(-1)^{\frac{r}{2}} (q-1) q^{m-\frac{r+2}{2}}, & \text{ if } r \text{ is even.}
  \end{array}
\right.
\end{align*}
Then
\begin{align*}
wt(f)=\frac{q^m-N(f=0)}{2}=
\left\{
  \begin{array}{ll}
     \frac{q^m-q^{m-1}}{2} , & \text{ if } r \text{ is odd};\\
     \frac{q^m-q^{m-1}  - \tau \eta(-1)^{\frac{r}{2}} (q-1) q^{m-\frac{r+2}{2}}}{2} , & \text{ if } r \text{ is even.}
  \end{array}
\right.
\end{align*}

When $a\neq 0$ and $\eta(-a)=1$, from  Lemma \ref{NQb}, we have
$$
N(f=0)=
\left\{
  \begin{array}{ll}
    q^{m-1}+\tau\eta(-1)^{\frac{r-1}{2}}
q^{m-\frac{r+1}{2}}, & \text{ if } r \text{ is odd};\\
    q^{m-1}  - \tau \eta(-1)^{\frac{r}{2}}  q^{m-\frac{r+2}{2}}, & \text{ if } r \text{ is even.}
  \end{array}
\right.
$$
Then
$$
wt(f)=\frac{q^m-1-N(f=0)}{2}=
\left\{
  \begin{array}{ll}
   \frac{1}{2}\left( q^m-q^{m-1}-\tau\eta(-1)^{\frac{r-1}{2}}
q^{m-\frac{r+1}{2}}-1 \right), & \text{ if } r \text{ is odd};\\
   \frac{1}{2}\left( q^m-q^{m-1}  + \tau \eta(-1)^{\frac{r}{2}}  q^{m-\frac{r+2}{2}}-1 \right), & \text{ if } r \text{ is even.}
  \end{array}
\right.
$$

When $a\neq 0$ and $\eta(-a)=-1$, from Lemma \ref{NQb}, we have
$$
N(f=0)=
\left\{
  \begin{array}{ll}
    q^{m-1}-\tau\eta(-1)^{\frac{r-1}{2}}
q^{m-\frac{r+1}{2}}, & \text{ if } r \text{ is odd};\\
    q^{m-1}  -\tau \eta(-1)^{\frac{r}{2}}  q^{m-\frac{r+2}{2}}, & \text{ if } r \text{ is even.}
  \end{array}
\right.
$$
Then
$$
wt(f)=\frac{q^m-1-N(f=0)}{2}=
\left\{
  \begin{array}{ll}
   \frac{1}{2}\left( q^m-q^{m-1}+\tau\eta(-1)^{\frac{r-1}{2}}
q^{m-\frac{r+1}{2}}-1 \right), & \text{ if } r \text{ is odd};\\
   \frac{1}{2}\left( q^m-q^{m-1} + \tau \eta(-1)^{\frac{r}{2}}  q^{m-\frac{r+2}{2}}-1 \right), & \text{ if } r \text{ is even.}
  \end{array}
\right.
$$
Hence, this lemma follows.
\end{proof}
\begin{lemma}\label{qurt}
Let $q$ be odd and $Q$ be a quadratic form of
rank $r\geq 1$ and type $\tau$ on
$\mathbb{F}_{q^m}$. Then the weight enumerator
$W_{r,\tau}$ of $\left(Q(\alpha^l) \right)_{l=0}^{n-1}$
is $Z^{\frac{q^m-q^{m-1}}{2}}$ (resp.
 $Z^{\frac{q^m-q^{m-1}  - \tau \eta(-1)^{\frac{r}{2}} (q-1) q^{m-\frac{r+2}{2}}} {2}})$ if $r$ is odd
(resp. even).
\end{lemma}
\begin{proof}
From results of $a=0$ in Lemma \ref{qurt-a}, this lemma follows.
\end{proof}

We determine the weight enumerators
of $\mathcal{C}_{1,h}$, $\tilde{\mathcal{C}}_{1,h}$, $\mathcal{C}_{1,h}$, and  $\tilde{\mathcal{C}}_{1,h}$ in the following theorem.
\begin{theorem}
When $m$ is odd, the weight enumerator of the code
$\mathcal{C}_{1,h}$   is
$$
1+(q-1)Z^{\frac{q^m-1}{2}}+\sum_{r=2h+1}^{m}\sum_{\tau\in \{1,-1\}}
a_{r,\tau}W_{r,\tau},
$$
where $a_{r,\tau}$ is given in  Theorem \ref{tart} and
$W_{r,\tau}$ is given in Lemma \ref{qurt-a}.
The weight enumerator of the code
$\tilde{\mathcal{C}}_{1,h}$   is
$$
1+\sum_{r=2h+1}^{m}\sum_{\tau\in \{1,-1\}}
a_{r,\tau}W_{r,\tau}
$$
where $a_{r,\tau}$ is given in  Theorem \ref{tart} and  $W_{r,\tau}$ is given in Lemma \ref{qurt}.
When $m$ is even, the weight enumerator of the code
$\mathcal{C}_{2,h}$   is
$$
1+(q-1)Z^{\frac{q^m-1}{2}}+
\sum_{r=2h}^{m}\sum_{\tau\in \{1,-1\}}
a_{r,\tau}W_{r,\tau}
$$
where $a_{r,\tau}$ is given in Proposition \ref{part} and
$W_{r,\tau}$ is given
in Lemma \ref{qurt-a}.
The weight enumerator of the code
$\tilde{\mathcal{C}}_{2,h}$   is
$$
1+\sum_{r=2h}^{m}\sum_{\tau\in \{1,-1\}}
a_{r,\tau}W_{r,\tau},
$$
where $a_{r,\tau}$ is given in Proposition \ref{part} and $W_{r,\tau}$ is given in Lemma \ref{qurt}.
\end{theorem}
\begin{proof}
We first give the weight enumerator of
$\mathcal{C}_{1,h}$. When $m$ is odd, we have
\begin{align*}
\mathcal{C}_{1,h}=
\bigcup_{Q\in Q_1}\left
(Q(\alpha^i)+a\right)_{l=0}^{n-1}.
\end{align*}
The weight enumerator of $\{(a)_{l=0}^{n-1}:
a\in \mathbb{F}_q\}$ is $1+(q-1)\frac{q^m-1}{2}$. From Proposition \ref{code-design}, $S_1$ is
a proper $(2h+1)$-code.
For any nonzero $Q(x)\in Q_1$, we have  $rank(Q(x))=r\geq 2h+1$.   The weight enumerator
$W_{r,\tau}$ of $\left(Q(\alpha^l)+a \right)_{l=0}^{n-1}$ with rank $r>0$ is given in Lemma \ref{qurt-a}. The inner distribution of $Q_1$
is just the inner distribution of $S_1$, which is given in  Theorem \ref{tart}. Hence, we have the
weight enumerator of $\mathcal{C}_{1,h}$
$$
1+(q-1)Z^{\frac{q^m-1}{2}}+\sum_{r=2h+1}^{m}\sum_{\tau\in \{1,-1\}}
a_{r,\tau}W_{r,\tau}.
$$
From the similar method, we have the other results. Hence, this theorem follows.
\end{proof}

\begin{lemma}\label{coset-leader}
When $m\geq 2$ and $1\leq i\leq  \lfloor \frac{m+3}{4}
\rfloor $, then
$\delta_i=\frac{q^m-q^{m-1}}{2}-1-\frac{q^{
\lfloor \frac{m-3}{2} \rfloor+i}-1}{2}$
is the $i$-th largest coset leader module $n$.
\end{lemma}
\begin{proof}
When $q=3$, from \cite{LDXG17}, this lemma holds.
Suppose that $q>3$.
When $m$ is odd,
then the $q$-adic expansion of
$\delta_i$ is
$$
\delta_i
=(\frac{q-3}{2})q^{m-1}
+(q-1)\sum_{j=\frac{m-3}{2}+i}^{m-2}
q^j+(\frac{q-1}{2})\sum_{j=0}^{\frac{m-3}{2}+i-1}q^j.
$$
and its cyclotomic coset $C_{\delta_i}$ is
$$
\left\{\frac{q^{m}-1-
q^{l-1}-q^{l+\frac{m-3}{2}+i}}{2}: 1\leq l\leq \frac{m+1}{2}-i   \right\}\bigcup
\left\{ \frac{q^{m}-1-
q^{l-1}-q^{l-\frac{m+3}{2}+i}}{2}: \frac{m+3}{2}-i\leq l\leq m \right\}.
$$
From Lemma 9 in \cite{ZSK19},
$\delta_1$ is the largest    coset leader module $n$. Then we just prove that for
$i> 1$, there does not exist a coset leader
$s$ satisfying $\delta_i< s< \delta_{i-1}$.

Suppose such a
coset leader $s$ satisfying $\delta_i< s< \delta_{i-1}$ exists.
From Lemma 8 in \cite{ZSK19},  then the $q$-adic expansion of $s$ is
$$
s=(\frac{q-3}{2})q^{m-1}
+(q-1)\sum_{j=k+1}^{m-2}
q^j+ \sum_{j=0}^{k}s_jq^j,
$$
where $k=\frac{m-3}{2}+i-1$, $\frac{q-3}{2}
\leq s_j\leq q-1$ and
$(\frac{q-1}{2})\sum_{j=0}^{k}q^j
<\sum_{j=0}^{k}s_jq^j<
(q-1)q^{k} +
(\frac{q-1}{2})\sum_{j=0}^{k-1}q^j
$.

When $s_k=\frac{q-1}{2}$,
from $t>\delta_i$, we have $(q^{m-2-t}s \mod n) < s$, where $0\leq t<k$ and $t$ is the largest index such that $s_t>\frac{q-1}{2}$. It is a contradiction to the coset leader $s$.

When  $\frac{q-1}{2}< s_k < q-1$, then
$(q^{m-1-k}s \mod n) < s$, which is a contradiction to the coset leader $s$.

When  $s_k=q-1$, from
$s<\delta_{i-1}$, there exists a $t$ satisfying   $(q^{m-1-t}s \mod n) < s$ or $(q^{m-2-t}s \mod n) < s$, where
$0\leq t<k$.  It is a contradiction to the coset leader $s$.

Hence, $\delta_i$ is the $i$-th largest coset leader module $n$ when $m$ is odd. From the similar method, it holds when $m$ is even. \end{proof}

The following theorem gives  parameters of the BCH code  $\mathcal{C}_{(n,q,m,\delta_i)}$.
\begin{theorem}
Let $m\geq 2$,  $1\leq i\leq  \lfloor \frac{m+3}{4}
\rfloor$,  and
$h=\lfloor \frac{m}{2} \rfloor -i+1$. Then the BCH code
$\mathcal{C}_{(n,q,m,\delta_i)}$  is
$\mathcal{C}_{1,h}$ (resp. $\mathcal{C}_{2,h}$) for
$m$ odd (resp. even) and it
has parameters
$[n,(i+1+\frac{m-1}{2}
-\lfloor \frac{m}{2} \rfloor)m,\delta_i]$ and
the  Bose distance $d_B=\delta_i$.
The code $\tilde{\mathcal{C}}_{(n,q,m,\delta_i)}$ is $\tilde{\mathcal{C}}_{1,h}$ (resp. $\tilde{\mathcal{C}}_{2,h}$)
for $m$ odd (resp. even) and it
has dimension
$(i+\frac{m-1}{2}
-\lfloor \frac{m}{2} \rfloor)m$.
\end{theorem}
\begin{proof}
Note that $\delta_i$ is the $i$-th largest coset leader module $n$.
When $m$ is odd, we have 
$l_i=\mid C_{\delta_i}\mid=m$ and
\begin{align*}
\mathrm{Tr}_1^{l_i}(b\beta^{-\delta_il})
=&\mathrm{Tr}_1^m\left(b\alpha^{\left(-q^m+
q^{m-1}+1+q^{\frac{m-3}{2}+i}\right)l}\right)\\
=& \mathrm{Tr}_1^m\left(b\alpha^{\left(
q^{m-1}+q^{\frac{m-3}{2}+i}\right)l}\right)\\
=& \mathrm{Tr}_1^m\left(b^{-\frac{m-3}{2}-i}\alpha^{
\left(q^{\frac{m+1}{2}-i}+1\right)l}\right).
\end{align*}
Hence the code
$\mathcal{C}_{(n,q,m,\delta_i)}=
\mathcal{C}_{1,h}$  and the
 dimension of $\mathcal{C}_{(n,q,m,\delta_i)}$ is
$(i+1+\frac{m-1}{2}
-\lfloor \frac{m}{2} \rfloor)m$.
When $r=2h+1$, the $\left(Q(\alpha^l)+a \right)_{l=0}^{n-1}$ has a codeword
of weight $\delta_i$.
From Lemma \ref{bose-d}, the code $\mathcal{C}_{(n,q,m,\delta_i)}$ has Bose distance $\delta_i$. Hence, the code $\mathcal{C}_{(n,q,m,\delta_i)}$
has parameters
$[n,(i+1+\frac{m-1}{2}
-\lfloor \frac{m}{2} \rfloor)m,\delta_i]$ and Bose distance $\delta_i$.
From the similar discussion, we have the other results of this proposition.
\end{proof}

\begin{example}
Let $q=3$, $h=\lfloor \frac{m}{2} \rfloor $, and
$i=1$.  When $m$ is odd (resp. even), the code  $\mathcal{C}_{(n,3,m,\delta_1)}$
is $\mathcal{C}_{1,\frac{m-1}{2}}$ (
resp. $\mathcal{C}_{2,\frac{m}{2}}$),   and
it has parameters
$[\frac{3^m-1}{2}, m+1, \delta_1]$
(resp. $[\frac{3^m-1}{2},
\frac{m+2}{2}, \delta_1]$ ) and the weight distribution in Table II (resp. Table I)
in Theorem 19 \cite{LDXG17}.
 When $m$ is odd (resp. even), the code  $\tilde{\mathcal{C}}_{(n,3,m,\delta_1)}$
is $\tilde{\mathcal{C}}_{1,\frac{m-1}{2}}$ (
resp. $\tilde{\mathcal{C}}_{2,\frac{m}{2}}$),   and
it has parameters
$[\frac{3^m-1}{2}, m,  3^{m-1}]$
(resp. $[\frac{3^m-1}{2},
\frac{m}{2},  3^{m-1}+ 3^{\frac{m-2}{2}}]$ ) and the weight enumerator $1+(3^m-1)Z^{3^{m-1}}$ (resp.  $1+(3^{\frac{m}{2}}-1 )Z^{3^{m-1}+3^{\frac{m-2}{2}}}$)
in Theorem 22 \cite{LDXG17}.
\end{example}
\begin{example}
Let $q=3$, $h=\lfloor \frac{m}{2} \rfloor-1 $, and
$i=2$.  When $m$ is odd (resp. even), the code  $\mathcal{C}_{(n,3,m,\delta_2)}$
is $\mathcal{C}_{1,\frac{m-3}{2}}$ (
resp. $\mathcal{C}_{2,\frac{m-2}{2}}$),   and
it has parameters
$[\frac{3^m-1}{2}, 2m+1, \delta_2]$
(resp. $[\frac{3^m-1}{2},
\frac{3m+2}{2}, \delta_2]$ ) and the weight distribution in Table VI (resp. Table V)
in Theorem 29 \cite{LDXG17}.
 When $m$ is odd (resp. even), the code  $\tilde{\mathcal{C}}_{(n,3,m,\delta_2)}$
is $\tilde{\mathcal{C}}_{1,\frac{m-3}{2}}$ (
resp. $\tilde{\mathcal{C}}_{2,\frac{m-2}{2}}$),   and
it has parameters
$[\frac{3^m-1}{2}, 2m,  3^{m-1}-3^{\frac{m-1}{2}}]$
(resp. $[\frac{3^m-1}{2},
\frac{3m}{2},  3^{m-1}- 3^{\frac{m-2}{2}}]$ ) and the weight distribution in Table IV (resp.  Table III)
in Theorem 26 \cite{LDXG17}.
\end{example}

\begin{example}
When $h=\lfloor \frac{m}{2} \rfloor $ and
$i=1$,  the code   $\mathcal{C}_{(n,q,m,\delta_1)}$
(resp. $\tilde{\mathcal{C}}_{(n,q,m,\delta_1)}$ )
is given in Theorem 5 (resp. Theorem 4) in
\cite{ZSK19}.
When $h=\lfloor \frac{m}{2} \rfloor-1 $ and
$i=2$,  the code   $\mathcal{C}_{(n,q,m,\delta_2)}$
(resp. $\tilde{\mathcal{C}}_{(n,q,m,\delta_2)}$ )
is given in Theorem 7 (resp. Theorem 6) in \cite{ZSK19}.
\end{example}
\begin{remark}
Let $q=3$, $h=\lfloor \frac{m}{2} \rfloor-1 $, and $i=2$. The code  $\tilde{\mathcal{C}}_{(n,3,m,\delta_2)}$ is a three-weight code with the weight distribution in Table 1 of \cite{TDX19}. Then these codes can be used to construct $2$-designs.
\end{remark}

\section{Conclusion}
This paper studies a class of linear codes of length $
\frac{q^m-1}{2}$ over $\mathbb{F}_q$ with special trace representation, uses association schemes from bilinear forms associated with quadratic forms, and determines the weight enumerators of these codes. From determining
some  cyclotomic coset leaders $\delta_i$ of
cyclotomic cosets modulo  $
\frac{q^m-1}{2}$, we prove that
narrow-sense BCH codes of length $
\frac{q^m-1}{2}$ with designed distance
$\delta_i=\frac{q^m-q^{m-1}}{2}-1-\frac{q^{
\lfloor \frac{m-3}{2} \rfloor+i}-1}{2}$ have
the corresponding trace representation and have the  minimal
distance $d=\delta_i$ and the Bose distance
$d_B=\delta_i$, where
$1\leq i\leq \lfloor \frac{m+3}{4} \rfloor$.
It is interesting to determine parameters of more BCH codes.

{\bf Acknowledgement.} 
This work was supported by SECODE project  and the National Natural Science Foundation of China
(Grant No. 11871058, 11531002, 11701129).  C. Tang also acknowledges support from
14E013, CXTD2014-4 and the Meritocracy Research Funds of China West Normal University. Y. Qi also acknowledges support from Zhejiang provincial Natural Science Foundation of China (LQ17A010008, LQ16A010005).

\end{document}